\documentclass[11pt]{article}
\usepackage{graphicx}
\usepackage{latexsym}
\usepackage{amssymb}
\usepackage{amsmath}
\usepackage{amsthm}
\usepackage{forloop}
\usepackage[paper=letterpaper,margin=1in]{geometry}
\usepackage[ruled,vlined,linesnumbered]{algorithm2e}
\usepackage{nicefrac}

\newtheorem{theorem}{Theorem}[section]
\newtheorem{lemma}[theorem]{Lemma}

\newtheorem{corollary}[theorem]{Corollary}

\newtheorem{observation}[theorem]{Observation}

\makeatletter
\newtheorem*{rep@theorem}{\rep@title}
\newcommand{\newreptheorem}[2]{%
\newenvironment{rep#1}[1]{%
 \def\rep@title{#2 \ref{##1}}%
 \begin{rep@theorem}}%
 {\end{rep@theorem}}}
\makeatother

\newreptheorem{theorem}{Theorem}
\newreptheorem{reduction}{Reduction}

\newcommand{\defcal}[1]{\expandafter\newcommand\csname c#1\endcsname{{\mathcal{#1}}}}
\newcommand{\defbb}[1]{\expandafter\newcommand\csname b#1\endcsname{{\mathbb{#1}}}}
\newcounter{calBbCounter}
\forLoop{1}{26}{calBbCounter}{
    \edef\letter{\Alph{calBbCounter}}
    \expandafter\defcal\letter
		\expandafter\defbb\letter
}

\newcommand{\eps}{\varepsilon}
\newcommand{\ie}{{\it i.e.}}
\newcommand{\eg}{{\it e.g.}}
\newcommand{\nnR}{{\bR_{\geq 0}}}
\newcommand{\characteristic}{{\mathbf{1}}}
\newcommand{\RSet}{{\mathtt{R}}}

\begin{document}

\title{\textbf{Constrained Submodular Maximization via a Non-symmetric Technique}}
\author
{
Niv	Buchbinder\thanks{Department of Statistics and Operations Research, School of Mathematical Sciences, Tel Aviv university, Israel. Email: \texttt{niv.buchbinder@gmail.com}.}
\and
Moran Feldman\thanks{Depart. of Mathematics and Computer Science, The Open University of Israel.
Email: \texttt{moranfe@openu.ac.il}.}
}
\maketitle

\begin{abstract}
The study of combinatorial optimization problems with a submodular objective has attracted much attention in recent years. Such problems are important in both theory and practice because their objective functions are very general. Obtaining further improvements for many submodular maximization problems boils down to finding better algorithms for optimizing a relaxation of them known as the multilinear extension.

In this work we present an algorithm for optimizing the multilinear relaxation whose guarantee improves over the guarantee of the best previous algorithm (which was given by Ene and Nguyen (2016)). Moreover, our algorithm is based on a new technique which is, arguably, simpler and more natural for the problem at hand. In a nutshell, previous algorithms for this problem rely on symmetry properties which are natural only in the absence of a constraint. Our technique avoids the need to resort to such properties, and thus, seems to be a better fit for constrained problems.
\end{abstract} 
\thispagestyle{empty}
\pagenumbering{Alph}
\newpage
\setcounter{page}{1}

\pagenumbering{arabic}

\section{Introduction}

The study of combinatorial optimization problems with a submodular objective has attracted much attention in recent years. Such problems are important in both theory and practice because their objective functions are very general---submodular functions generalize, for example, cuts functions of graphs and directed graphs, the mutual information function, matroid weighted rank functions and log-determinants. More specifically, from a theoretical perspective, many well-known problems in combinatorial optimization are in fact submodular maximization problems, including: Max-Cut~\cite{GW95,H01,K72,KKMO07,TSSW00}, Max-DiCut~\cite{FG95,GW95,HZ01}, Generalized Assignment~\cite{CK05,CKR06,FV06,FGMS06}, Max-$k$-Coverage~\cite{F98,KMN99}, Max-Bisection~\cite{ABG13,FJ95} and Facility Location~\cite{AS99,CFN77a,CFN77b}. From a practical perspective, submodular maximization problems have found uses in social networks~\cite{HMS08,KKT03}, vision~\cite{BJ01,JB11}, machine learning~\cite{KSG08,KG05,KLGVF08,LB10,LB11} and many other areas (the reader is referred, for example, to a comprehensive survey by Bach~\cite{Bach13}).

The techniques used by approximation algorithms for submodular maximization problems usually fall into one of two main approaches. The first approach is combinatorial in nature, and is mostly based on local search techniques and greedy rules. This approach has been used as early as the late $70$'s for maximizing a monotone submodular function subject to a matroid constraint (some of these works apply only to specific types of matroids)~\cite{CC84,FNW78,HK78,HKJ80,J76,KH78,NW78,NWF78}. Later works used this approach to handle also problems with non-monotone submodular objective functions and different constraints~\cite{BFNS12,FMV11,FNSW11,LMNS10,LSV09}, yielding in some cases optimal algorithms~\cite{BFNS12,M04}. However, algorithms based on this approach tend to be highly tailored for the specific structure of the problem at hand, making extensions quite difficult.

The second approach used by approximation algorithms for submodular maximization problems overcomes the above obstacle. This approach resembles a common paradigm for designing approximation algorithms and involves two steps. In the first step a fractional solution is found for a relaxation of the problem, known as the \emph{multilinear relaxation}. In the second step the fractional solution is rounded to obtain an integral one while incurring a bounded loss in the objective. This approach has been used to obtain improved approximations for many problems~\cite{CCPV11,CVZ10,CVZ11,FNS11,KST13}.
%Any work based on this approach must address two issues. First, one needs to (approximately) optimize the multilinear relaxation. Second, given a fractional solution for this relaxation, one needs a rounding procedure which outputs an integral solution without losing too much in the objective function.

Various techniques have been developed for rounding the fractional solution. These techniques tend to be quite flexible, and usually can extend to many related problem. In particular, the Contention Resolution Schemes framework of~\cite{CVZ11} yields a rounding procedure for every constraint which can be presented as the intersection of a few basic constraints such as knapsack constraints, matroid constraints and matching constraints. Given this wealth of rounding procedures, obtaining further improvements for many important submodular maximization problems (such as maximizing a submodular function subject to a matroid or knapsack constraint) boils down to obtaining improved algorithms for finding a good fractional solution, \ie, optimizing the multilinear relaxation.

%Various techniques have been developed by the above works in order to address the second issue, namely, rounding the fractional solution. These techniques tend to be quite flexible, and usually can extend to many related problem. In particular, the Contention Resolution Schemes framework of~\cite{CVZ11} yields a rounding procedure for every constraint which can be presented as the intersection of a few basic constraints such as knapsack constraints, matroid constraints and matching constraints. Given this wealth of rounding procedures, obtaining further improvements for many important submodular maximization problems (such as maximizing a submodular function subject to a matroid or knapsack constraint) boils down to finding improved algorithms for tackling the first issue, \ie, optimizing the multilinear extension.

\subsection{Maximizing the Multilinear Relaxation}

At this point we would like to present some terms more formally. A submodular function is a set function $f\colon 2^\cN \to \bR$ obeying $f(A) + f(B) \geq f(A \cup B) + f(A \cap B)$ for any sets $A, B \subseteq \cN$.
A submodular maximization problem is the problem of finding a set $S \subseteq \cN$ maximizing $f$ subject to some constraint. Formally, let $\cI$ be the set of subsets of $\cN$ obeying the constraint. Then, we are interested in the following problem.
\[ \begin{array}{ll}
	\max & f(A) \\
	\text{s.t.} & A \in \cI \subseteq 2^\cN
\end{array} \]

%For example, the constraint may require $S$ to contain at most $k$ elements, for some value $k$.

A relaxation of the above problem replaces $\cI$ with a polytope $P \subseteq [0, 1]^\cN$ containing the characteristic vectors of all the sets of $\cI$. In addition, a relaxation must replace the function $f$ with an extension function $F\colon [0,1]^\cN \to \bR$. Thus, a relaxation is a fractional problem of the following format.

\[ \begin{array}{ll}
	\max & F(x) \\
	\text{s.t.} & x \in P \subseteq [0, 1]^\cN
\end{array} \]

Defining the ``right'' extension function, $F$, for the relaxation is a challenge, as, unlike the linear case, there is no single natural candidate. The objective that turned out to be useful, and is, thus, used by multilinear relaxation is known as the \emph{multilinear extension} (first introduced by \cite{CCPV11}). The value $F(x)$ of this extension for any vector $x \in [0,1]^\cN$ is defined as the expected value of $f$ over a random subset $\RSet(x) \subseteq \cN$ containing every element $u \in \cN$ independently with probability $x_u$. Formally, for every $x\in [0,1]^\cN$,
\[
	F(x)
	=
	\bE[\RSet(x)]
	=
	\sum_{S\subseteq \cN} f(S) \prod _{u\in S}x_u \prod _{u\notin S}(1-x_u)
	\enspace.
\]

%The multilinear relaxation of such a problem, like standard linear relaxations, replace the constraint with a polytope $P \subseteq [0, 1]^\cN$.
%Finding an objective function for the relaxation is more involved, as, unlike in the linear case, there is no single natural candidate. The objective that turned out to be useful, and is, thus, used by multilinear relaxation is an extension of $f$, known as the \emph{multilinear extension}, which was first introduced by \cite{CCPV11}. The value $F(x)$ of this extension for any vector $x \in [0,1]^\cN$ is defined as the expected value of $f$ over a random subset $\RSet(x) \subseteq \cN$ containing every element $u \in \cN$ independently with probability $x_u$. Formally, for every $x\in [0,1]^\cN$,
%\[
%	F(x)
%	=
%	\bE[\RSet(x)]
%	=
%	\sum_{S\subseteq \cN} f(S) \prod _{u\in S}x_u \prod _{u\notin S}(1-x_u)
%	\enspace.
%\]

The first algorithm for optimizing the multilinear relaxation was the Continuous Greedy algorithm designed by Calinescu et al.~\cite{CCPV11}. When the submodular function $f$ is non-negative and monotone\footnote{A set function $f\colon 2^\cN \to \bR$ is monotone if $f(A) \leq f(B)$ for every $A \subseteq B \subseteq \cN$.} and $P$ is solvable\footnote{A polytope is solvable if one can optimize linear functions over it.} this algorithm finds a vector $x \in P$ such that $\bE[F(x)] \geq (1 - \nicefrac{1}{e} - o(1)) \cdot f(OPT)$ (where $OPT$ is the set maximizing $f$ among all sets whose characteristic vectors belongs to $P$).
%\footnote{\Bnote{I would remove this} In fact, Continuous Greedy guarantees $F(x) \geq (1 - e^{-1} - o(1)) \cdot f(OPT)$ with high probability. However, the two guarantees are equivalent since any algorithm guaranteeing $\bE[F(x)] \geq (c - o(1)) \cdot f(OPT)$ for some constant $c$ can be made to guarantee $F(x) \geq (c - o(1)) \cdot f(OPT)$ with high probability by simply repeating the original algorithm a logarithmic number of times, and then returning the best solution found.}
Interestingly, the guarantee of Continuous Greedy is optimal for monotone functions even when $P$ is a simple cardinality constraint~\cite{CCPV11,NW78}.

Optimizing the multilinear relaxation when $f$ is not necessarily monotone proved to be a more challenging task. Initially, several algorithms for specific polytopes were suggested~\cite{GV11,LMNS10,V13}.
Later on, improved general algorithms were designed that work whenever $f$ is non-negative and $P$ is down-closed\footnote{A polytope $P \subseteq [0, 1]^\cN$ is down-closed if $y \in P$ implies that every vector $x \in [0, 1]^\cN$ which is coordinate-wise upper bounded by $y$ must belong to $P$ as well.} and solvable~\cite{CVZ14,FNS11}. Designing algorithms that work in this general setting is highly important as many natural constraints fall into this framework. Moreover, the restriction of the algorithms to down-closed polytopes is unavoidable as Vondr\'{a}k~\cite{V13} proved that no algorithm can produce a vector $x \in P$ obeying $\bE[F(x)] \geq c \cdot f(OPT)$ for any constant $c > 0$ when $P$ is solvable but not down-closed.

%and were later improved by algorithms that work whenever $f$ is non-negative and $P$ is down-closed\footnote{A polytope $P \subseteq [0, 1]^\cN$ is down-closed if $y \in P$ implies that every vector $x \in [0, 1]^\cN$ which is coordinate-wise upper bounded by $y$ must belong to $P$ as well.} and solvable~\cite{CVZ14,FNS11}.

%last kind was an algorithm of~\cite{FNS11} called Measured Continuous Greedy which is guaranteed to produce a vector $x \in P$ obeying $\bE[F(x)] \geq (e^{-1} - o(1)) \cdot f(OPT) \approx 0.367 \cdot f(OPT)$.

Up until recently, the best algorithm for this general setting was called Measured Continuous Greedy~\cite{FNS11}. It guaranteed to produce a vector $x \in P$ obeying $\bE[F(x)] \geq (\nicefrac{1}{e} - o(1)) \cdot f(OPT) \approx 0.367 \cdot f(OPT)$~\cite{FNS11}. The natural feel of the guarantee of Measured Continuous Greedy and the fact that it was not improved for a few years made some people suspect that it is optimal. Recently, an evidence against this conjecture was given by~\cite{BFNS14}, which described an algorithm for the special case of a cardinality constraint with an improved approximation guarantee of $0.371$. Even more recently, Ene and Nguyen~\cite{EN16} shuttered the conjecture completely. By extending the technique used by~\cite{BFNS14}, they showed that one can get an approximation guarantee $0.372$ for every down-closed and solvable polytope $P$. On the inapproximability side, Oveis Gharan and Vondr\'{a}k~\cite{GV11} proved that no algorithm can achieve approximation better than $0.478$  even when $P$ is the matroid polytope of a partition matroid. Closing the gap between the best algorithm and inapproximability result for this fundamental problem remains an important open problem.

%\Bnote{In additional related: and Vondr\'{a}k~\cite{V13} proved that no algorithm can produce a vector $x \in P$ obeying $\bE[F(x)] \geq c \cdot f(OPT)$ for any constant $c > 0$ when $P$ is solvable but not down-closed.}

%\Bnote{We don't want to say this: using a technique very similar to the one used by~\cite{BFNS14},}

\subsection{Our Contribution}

Our main contribution is an algorithm with an improved guarantee for maximizing the multilinear relaxation.

\begin{theorem} \label{thm:main_theorem}
There exists a polynomial time algorithm that given a non-negative submodular function $f\colon 2^\cN \to \nnR$ and a solvable down-closed polytope $P \subseteq [0, 1]^\cN$ finds a vector $x \in P$ obeying $F(x) \geq 0.385 \cdot f(OPT)$, where $OPT = \arg \max \{f(S) : 1_S \in P\}$ and $F$ is the multilinear extension of $f$.
\end{theorem}

Admittedly, the improvement in the guarantee obtained by our algorithm compared to the $0.372$ guarantee of~\cite{EN16} is relatively small. However, the technique underlying our algorithm is very different, and, arguably, much cleaner, than the technique underlying the previous results improving over the natural guarantee of $\nicefrac{1}{e}$~\cite{BFNS14,EN16}. Moreover, we believe our technique is more natural for the problem at hand, and thus, is likely to yield further improvements in the future. In the rest of this section we explain the intuition on which we base this belief.

The results of~\cite{BFNS14,EN16} are based on the observation that the guarantee of Measured Continuous Greedy improves when the algorithm manages to increase all the coordinates of its solution at a slow rate. Based on this observation, \cite{BFNS14,EN16} run an instance of Measured Continuous Greedy (or a discretized version of it), and force it to raise the coordinates slowly. If this extra restriction does not affect the behavior of the algorithm significantly, then it produces a solution with an improved guarantee. Otherwise, \cite{BFNS14,EN16} argue that the point in which the extra restriction affect the behavior of Measured Continuous Greedy reveals a vector $x \in P$ which contains a significant fraction of $OPT$. Once $x$ is available, one can use the technique of unconstrained submodular maximization, described by~\cite{BFNS12}, that has higher approximation guarantee of $\nicefrac{1}{2}> \nicefrac{1}{e}$, to extract from $x$ a vector $0 \leq y \leq x$ of large value. The down-closeness of $P$ guarantees that $y$ belongs to $P$ as well.

Unfortunately, the use of the unconstrained submodular maximization technique in the above approach is very problematic for two reasons. First, this technique is based on ideas that are very different from the ideas used by the analysis of Measured Continuous Greedy. This makes the combination of the two quite involved. Second, on a more abstract level, the unconstrained submodular maximization technique is based on a symmetry which exists in the absence of a constraint since $\bar{f}(S) = f(\cN \setminus S)$ is non-negative and submodular whenever $f$ has these properties. However, this symmetry breaks when a constraint is introduced, and thus, the unconstrained submodular maximization technique does not seem to be a good fit for a constrained problem.

Our algorithm replaces the symmetry based unconstrained submodular maximization technique with a local search algorithm. More specifically, it first executes the local search algorithm. If the output of the local search algorithm is good, then our algorithm simply returns it. Otherwise, we observe that the poor value of the output of the local search algorithm guarantees that it is also far from $OPT$ in some sense. Our algorithm then uses this far from $OPT$ solution to guide an instance of Measured Continuous Greedy, and help it avoid bad decisions.

As it turns out, the analysis of Measured Continuous Greedy and the local search algorithm use similar ideas and notions. Thus, the two algorithms combine quite cleanly, as can be observed from Section~\ref{sec:main_algorithm}. 
\section{Preliminaries}

Our analysis uses another useful extension of submodular functions. Given a submodular function $f\colon 2^\cN \to \bR$, its Lov\'{a}sz extension is a function $\hat{f} \colon [0, 1]^\cN \to \bR$ defined by
\[
	\hat{f}(x) = \int_0^1 f(T_\lambda(x)) d\lambda
	\enspace,
\]
where $T_\lambda(x) = \{u \in \cN : x_u < \lambda\}$. The Lov\'{a}sz extension has many important applications (see, \eg,~\cite{CE11,GLS81}), however, in this paper we only use it in the context of the following known result (which is an immediate corollary of the work of~\cite{L83}).

\begin{lemma} \label{lem:extensions_relation}
Given the multilinear extension $F$ and the Lov\'{a}sz extension $\hat{f}$ of a submodular function $f \colon 2^\cN \to \bR$, it holds that $F(x) \geq \hat{f}(x)$ for every vector $x \in [0, 1]^\cN$.
\end{lemma}

We now define some additional notation that we use. Given a set $S \subseteq \cN$ and an element $u \in \cN$, we denote by $\characteristic_S$ and $\characteristic_u$ the characteristic vectors of the sets $S$ and $\{u\}$, respectively, and by $S + u$ and $S - u$ the sets $S \cup \{u\}$ and $S \setminus \{u\}$, respectively. Given two vectors $x, y \in [0, 1]^\cN$, we denote by $x \vee y$, $x \wedge y$ and $x \circ y$ the coordinate-wise maximum, minimum and multiplication, respectively, of $x$ and $y$.\footnote{More formally, for every element $u \in \cN$, $(x \vee y)_u = \max\{x_u, y_u\}$, $(x \wedge y)_u = \min\{x_u, y_u\}$ and $(x \circ y)_u = x_u \cdot y_u$.} Finally, given a vector $x \in [0, 1]^\cN$ and an element $u \in \cN$, we denote by $\partial_u F(x)$ the derivative of $F$ with respect to $u$ at the point $x$. The following observation gives a simple formula for $\partial_u F(x)$. This observation holds because $F$ is a multilinear function.
\begin{observation}
Let $F(x)$ be the multilinear extension of a submodular function $f\colon 2^\cN \to \bR$. Then, for every $u \in \cN$ and $x \in [0, 1]^\cN$,
\[
    (1-x_u) \cdot \partial_u F(x)
		=
    F(x \vee \characteristic_u) - F(x)
		\enspace.
\]
\end{observation}

In the rest of the paper we assume, without loss of generality, that $\characteristic_{u} \in P$ for every element $u \in \cN$ and that $n$ is larger than any given constant. The first assumption is justified by the observation that every element $u$ violating this assumption can be safely removed from $\cN$ since it cannot belong to $OPT$. The second assumption is justified by the observation that it is possible to find a set $S$ obeying $\characteristic_S \in P$ and $f(S) = f(OPT)$ in constant time when $n$ is a constant.

Another issue that needs to be kept in mind is the representation of submodular functions. We are interested in algorithms whose time complexity is polynomial in $|\cN|$. However, the representation of the submodular function $f$ might be exponential in this size; thus, we cannot assume that the representation of $f$ is given as part of the input for the algorithm. The standard way to bypass this difficulty is to assume that the algorithm has access to $f$ through an oracle. We assume the standard \emph{value oracle} that is used in most of the previous works on submodular maximization. This oracle returns, given any subset $S \subseteq \cN$, the value $f(S)$. 
\section{Main Algorithm} \label{sec:main_algorithm}

In this section we present the algorithm used to prove Theorem~\ref{thm:main_theorem}. This algorithm uses two components. The first component is a close variant of a fractional local search algorithm suggested by Chekuri et al.~\cite{CVZ14} which has the following properties.

\begin{lemma}[Follows from Chekuri et al.~\cite{CVZ14}] \label{cor:old_component_tailored}
There exists a polynomial time algorithm which returns vector $x\in P$ such that, with high probability, for every vector $y \in P$,
\begin{equation}\label{ineq:auxiliary1}
	F(x) \geq \frac{1}{2}F(x \wedge y) + \frac{1}{2}F(x \vee y) - o(1) \cdot f(OPT) %\frac{5 \cdot f(OPT)}{2n}
	\enspace.
\end{equation}
%with high probability, terminates after $T(n)$ operations, where $T(n)$ is a polynomial function of $n$, and outputs a vector $x \in P$ obeying
%\[
%	2F(x)
%	\geq
%	F(x \wedge y) + F(x \vee y) - \frac{5 \cdot f(OPT)}{n}
%\]
%for every vector $y \in P$. Moreover, the output vector $x$ belongs to $P$ whenever the algorithm terminates.
\end{lemma}
\begin{proof}
Let $M = \max\{f(u), f(\cN - u) : u \in \cN\}$, and let $a$ be an arbitrary constant larger than $3$. Then, Lemmata~3.7 and 3.8 of Chekuri et al.~\cite{CVZ14} imply that, with high probability, the fractional local search algorithm they suggest terminates in polynomial time and outputs a vector $x \in P$ obeying, for every vector $y \in P$,
\[
	2F(x)
	\geq
	F(x \wedge y) + F(x \vee y) - \frac{5M}{n^{a - 2}}
	\enspace.
\]
Moreover, the output vector $x$ is in $P$ whenever the fractional local search algorithm terminates.

Our assumption that $\characteristic_{u} \in P$ for every element $u \in \cN$ implies, by submodularity, that $f(S) \leq n \cdot f(OPT)$ for every set $S \subseteq \cN$. Since $M$ is the maximum over values of $f$, we get also $M \leq n \cdot f(OPT)$. Using this observation, and plugging $a = 4$, we get that there exists an algorithm which, with high probability, terminates after $T(n)$ operations (for some polynomial function $T(n)$) and outputs a vector $x \in P$ obeying
$2F(x) \geq F(x \wedge y) + F(x \vee y) - \frac{5 \cdot f(OPT)}{n}$ for every vector $y \in P$. Moreover, the output vector $x$ belongs to $P$ whenever the algorithm terminates.

To complete the lemma, we consider a procedure that executes the above algorithm for $T(n)$ operations, and return its output if it terminates within this number of operations. If the algorithm fails to terminate within this number of operations, which happens with a diminishing probability, then the procedure simply returns $1_{\varnothing}$ (which always belongs to $\cP$ since $\cP$ is down-closed). One can observe that this procedure has all the properties guaranteed by the lemma.
\end{proof}

%\begin{lemma}[rephrasing of Lemmata~3.7 and 3.8 of Chekuri et al.~\cite{CVZ14}] \label{lem:old_component}
%Let $M = \max\{f(u), f(\cN - u) : u \in \cN\}$, and let $a$ be an arbitrary constant larger than $3$. Then, with high probability, the fractional local search algorithm terminates in polynomial time and outputs a vector $x \in P$ obeying
%\[
%	2F(x)
%	\geq
%	F(x \wedge y) + F(x \vee y) - \frac{5M}{n^{a - 2}}
%\]
%for every vector $y \in P$. Moreover, the output vector $x$ belongs to $P$ whenever the fractional local search algorithm terminates.
%\end{lemma}

%Our assumption that $\characteristic_{u} \in P$ for every element $u \in \cN$ implies, by submodularity, that $f(S) \leq n \cdot f(OPT)$ for every set $S \subseteq \cN$. Since $M$ is the maximum over values of $f$, we get also $M \leq n \cdot f(OPT)$. Using this observation, the next corollary follows from Lemma~\ref{lem:old_component} by plugging $a = 4$.
%
%\begin{corollary} \label{cor:old_component_tailored}
%There exists an algorithm which, with high probability, terminates after $T(n)$ operations, where $T(n)$ is a polynomial function of $n$, and outputs a vector $x \in P$ obeying
%\[
%	2F(x)
%	\geq
%	F(x \wedge y) + F(x \vee y) - \frac{5 \cdot f(OPT)}{n}
%\]
%for every vector $y \in P$. Moreover, the output vector $x$ belongs to $P$ whenever the algorithm terminates.
%\end{corollary}

The second component of our algorithm is a new auxiliary algorithm which we present and analyze in Section~\ref{sec:auxiliary_algorithm}. This auxiliary algorithm is the main technical contribution of this paper, and its guarantee is given by the following theorem.
\begin{theorem} \label{thm:auxiliary_algorithm}
There exists a polynomial time algorithm that given a vector $z \in [0, 1]^\cN$ and a value $t_s \in [0, 1]$ outputs a vector $x \in P$ obeying
\begin{align}\label{ineq:auxiliary2}
	\bE[F(x)]
	\geq
	e^{t_s - 1} \cdot [(2 - t_s - e^{-t_s} - o(1)) \cdot f(OPT) &- (1 - e^{-t_s}) \cdot F(z \wedge \characteristic_{OPT})\\&- (2 - t_s - 2e^{-t_s}) \cdot F(z \vee \characteristic_{OPT})]
	\enspace. \nonumber
\end{align}
\end{theorem}

%We are now ready to present the algorithm used to prove Theorem~\ref{thm:main_theorem}.
Our main algorithm executes the algorithms suggested by Lemma~\ref{cor:old_component_tailored} followed by the algorithm suggested by Theorem~\ref{thm:auxiliary_algorithm}. Notice that the second of these algorithms has two parameters in addition to $f$ and $P$: a parameter $z$ which is set to be the output of the first algorithm, and a parameter $t_s$ which is set to be a constant to be determined later. After the two above algorithms terminate, our algorithm returns the output of the first algorithm with probability $p$, for a constant $p$ to be determined later, and with the remaining probability it returns the output of the second algorithm.\footnote{Clearly it is always better to return the better of the two solution instead of randomizing between them. However, doing so will require the algorithm to either have an oracle access to $F$ or estimate the values of the solutions using sampling (the later can be done using standard techniques---see, \eg,~\cite{CCPV11}). For the sake of simplicity, we chose here the easier to analyze approach of randomizing between the two solutions.} A formal description of our algorithm is given as Algorithm~\ref{alg:main_algorithm}. Observe that Lemma~\ref{cor:old_component_tailored} and Theorem~\ref{thm:auxiliary_algorithm} imply together that Algorithm~\ref{alg:main_algorithm} is a polynomial time algorithm which always outputs a vector in $P$.

%Our main algorithm starts by executing the algorithm suggested by Corollary~\ref{cor:old_component_tailored}. This algorithm only stops in polynomial time with high probability, and thus, if it runs for too long our algorithm stops it and outputs the legal (but probably very bad) solution $\characteristic_\varnothing$. On the other hand, if the algorithm suggested by Corollary~\ref{cor:old_component_tailored} stops on time and produces an output, then our algorithm runs also the algorithm suggested by Theorem~\ref{thm:auxiliary_algorithm} and uses the output of the first algorithm as an input for the second algorithm. In this case the output of our algorithm is the better solution among the outputs of the two algorithms it executes. A formal description of our algorithm is given as Algorithm~\ref{alg:main_algorithm}.

\begin{algorithm}
\caption{\textsf{Main Algorithm}($f, P$)} \label{alg:main_algorithm}
\DontPrintSemicolon
Execute the algorithm suggested by Lemma~\ref{cor:old_component_tailored}, and let $x_1\in P$ be its output. \\
Execute the algorithm suggested by Theorem~\ref{thm:auxiliary_algorithm} with $z = x_1$, and let $x_2$ be its output.\\
\Return{with probability $p$ the solution $x_1$, and the solution $x_2$ otherwise.}
%Execute the algorithm suggested by Lemma~\ref{cor:old_component_tailored}, and let it make $T(n)$ operations.\\
%\lIf{the algorithm suggested by Corollary~\ref{cor:old_component_tailored} did not terminate after $T(n)$ operations}{\Return{$\characteristic_\varnothing$.}}
%\Else
%{
%	Let $x_1$ be the output of the algorithm suggested by Corollary~\ref{cor:old_component_tailored}.\\
%	Execute the algorithm suggested by Theorem~\ref{thm:auxiliary_algorithm} with $z = x_1$ and $t_s = 0.4$, and let $x_2$ denote its output.\\
%	\Return{with probability $\frac{0.324}{1.324}$ the solution $x_1$, and the solution $x_2$ otherwise.\footnotemark}
%}
\end{algorithm}

%Let us begin the analysis of Algorithm~\ref{alg:main_algorithm} by proving the following easy observation.
%\begin{observation} \label{obs:feasibility}
%Algorithm~\ref{alg:main_algorithm} is a polynomial time algorithm which always outputs a vector inside the polytope $P$
%\end{observation}
%\begin{proof}
%The polynomial time complexity of Algorithm~\ref{alg:main_algorithm} follows from the fact that the algorithm suggested by Theorem~\ref{thm:auxiliary_algorithm} has a polynomial time complexity, and the algorithm suggested by Corollary~\ref{cor:old_component_tailored} is only allowed by Algorithm~\ref{alg:main_algorithm} to run for a polynomial time.
%
%It remains to prove that Algorithm~\ref{alg:main_algorithm} always outputs a vector inside $P$. %Notice that the algorithm always outputs one of the vectors $x_1$, $x_2$ or $\characteristic_\varnothing$.
%Corollary~\ref{cor:old_component_tailored} and Theorem~\ref{thm:auxiliary_algorithm} guarantee that both $x_1$ and $x_2$ belong to $P$ whenever the algorithms calculating them terminate. Moreover, the vector $\characteristic_\varnothing$ also belongs to $P$ since $P$ is down-closed. Hence, the vector outputted by Algorithm~\ref{alg:main_algorithm}, which is always one of these three vectors, must belong to $P$.
%\end{proof}

To prove Theorem~\ref{thm:main_theorem}, it remains to analyze the quality of the solution produced by Algorithm~\ref{alg:main_algorithm}.

\begin{lemma}
When its parameters are set to $t_s=0.372$ and $p=0.23$, Algorithm~\ref{alg:main_algorithm} produces a solution whose expected value is at least $0.385 \cdot f(OPT)$.
\end{lemma}
\begin{proof}
Let $\cE$ be the event that $x_1$, the output of the algorithm suggested by Lemma~\ref{cor:old_component_tailored}, satisfies Inequality~\eqref{ineq:auxiliary1}. Since $\cE$ is a high probability event, it is enough to prove that, conditioned on $\cE$, Algorithm~\ref{alg:main_algorithm} produces a solution whose expected value is at least $c \cdot f(OPT)$ for some constant $c > 0.385$. The rest of the proof of the lemma is devoted to proving the last claim. Throughout it, everything is implicitly conditioned on $\cE$.

As we are conditioning on $\cE$, we can plug $y = \characteristic_{OPT}$ and, respectively, $y = x_1 \wedge \characteristic_{OPT}$ into Inequality~\eqref{ineq:auxiliary1} to get
\begin{equation}
	F(x_1)
	\geq
	\frac{1}{2} F(x_1 \wedge \characteristic_{OPT}) + \frac{1}{2} F(x_1 \vee \characteristic_{OPT}) - o(1) \cdot f(OPT) \label{guarantee1}
\end{equation}
and
\begin{equation}
	F(x_1)
	\geq
	F(x_1 \wedge \characteristic_{OPT}) - o(1) \cdot f(OPT) \label{guarantee2}
	\enspace,
\end{equation}
where the last inequality follows by noticing that $x_1 \vee (x_1 \wedge \characteristic_{OPT})= x_1$.
Next, let $\bE[F(x_2) \mid x_1]$ denote the expected value of $F(x_2)$ conditioned on the given value of $x_1$.
Inequality \eqref{ineq:auxiliary2} guarantees that
\begin{align}\label{guarantee3}
	\bE[F(x_2) \mid x_1]
	\geq
	e^{t_s - 1} \cdot [(2 - t_s - e^{-t_s} - o(1)) \cdot f(OPT) &- (1 - e^{-t_s}) \cdot F(x_1 \wedge \characteristic_{OPT}) \\&- (2 - t_s - 2e^{-t_s}) \cdot F(x_1 \vee \characteristic_{OPT})]
	\enspace.
	\nonumber
\end{align}

Recall that Algorithm~\ref{alg:main_algorithm} returns $x_1$ with probability $p$, and $x_2$ otherwise. Hence, the expected value of its output is
\begin{equation} \label{eq:approximation_ratio}
	\bE[p \cdot F(x_1) + (1 - p) \cdot \bE[F(x_2) \mid x_1]]
	\enspace,
\end{equation}
where the expectation is over $x_1$.

\paragraph{Optimizing the constants.} We would like to derive from Inequalities~\eqref{guarantee1}, \eqref{guarantee2} and~\eqref{guarantee3} the best lower bound we can get on~\eqref{eq:approximation_ratio}. To this end, let $p_1$ and $p_2$ be two non-negative numbers such that $p_1 + p_2 = p$, and let $p_3 = 1 - p$. Using the above inequalities and this notation, \eqref{eq:approximation_ratio} can now be lower bounded by
\begin{alignat*}{3}
	&
	p_1 \cdot \left[\frac{1}{2} \bE[F(x_1 \wedge \characteristic_{OPT})] + \frac{1}{2} \bE[F(x_1 \vee \characteristic_{OPT})] - o(1) \cdot f(OPT)\right]\mspace{-200mu}\\
	{}+{}&
	p_2 \cdot \left[\bE[F(x_1 \wedge \characteristic_{OPT})] - o(1) \cdot f(OPT)\right]\\
	{}+{}&
	p_3 \cdot e^{t_s - 1} \cdot [(2 - t_s - e^{-t_s} - o(1)) \cdot f(OPT) &&- (1 - e^{-t_s}) \cdot \bE[F(x_1 \wedge \characteristic_{OPT})] \\&&&- (2 - t_s - 2e^{-t_s}) \cdot \bE[F(x_1 \vee \characteristic_{OPT})]]
	\enspace,
\end{alignat*}
which can be rewritten as
\begin{align*}
	& \left(\frac{p_1}{2} + p_2 - p_3\cdot e^{t_s-1}(1 - e^{-t_s})\right) \cdot \bE[F(x_1 \wedge \characteristic_{OPT})] \\
	{}+{} & \left(\frac{p_1}{2} - p_3\cdot e^{t_s-1}(2 - t_s - 2e^{-t_s})\right) \cdot \bE[F(x_1 \vee \characteristic_{OPT})] \\
	{}+{} & p_3 \cdot e^{t_s - 1}(2 - t_s - e^{-t_s}) \cdot f(OPT) - o(1) \cdot f(OPT)
	\enspace.
\end{align*}

To get the most out of this lower bound we need to maximize the coefficient of $f(OPT)$ while keeping the coefficients of $\bE[F(x_1 \wedge \characteristic_{OPT})]$ and $\bE[F(x_1 \vee \characteristic_{OPT})]$ non-negative (so that they can be ignored due to non-negativity of $f$). This objective is formalized by the following non-convex program.

\[\begin{array}{lll}
	\max & p_3 \cdot e^{t_s - 1}(2 - t_s - e^{-t_s})\\
	\text{s.t.} & p_1/2 + p_2 - p_3 \cdot e^{t_s - 1}(1 - e^{-t_s})  & \geq 0\\
	& p_1/2 - p_3 \cdot e^{t_s - 1}(2 - t_s - 2e^{-t_s}) & \geq 0\\
	& p_1 + p_2 + p_3 & = 1\\
	& p_1, p_2, p_3, t_s & \geq 0
\end{array}\]

Solving the program, we get that the best solution is approximately $p_1 = 0.205$, $p_2 = 0.025$, $p_3 = 0.770$ and $t_s = 0.372$, and the objective function value corresponding to this solution is at least $0.3856$. Hence, we have managed to lower bound \eqref{eq:approximation_ratio} (and thus, also the expected value of the output of Algorithm~\ref{alg:main_algorithm}) by $0.3856 \cdot f(OPT)$ for $p = 0.23$ and $t_s = 0.372$, which completes the proof of the lemma.
%
%\Bnote{
%The approximation ratio is $1-\frac{1}{W(\frac{e^{1+1/(e-2)}}{e-2})+1-\frac{1}{e-2}}$, $c=\frac{1}{e-2}$, $1-\frac{1}{W(ce^{1+c})+1-c}$. $W()$ is the Product log is the principal solution for $w$ in $z=we^w$.}
\end{proof}

\section{Aided Measured Continuous Greedy} \label{sec:auxiliary_algorithm}

In this section we present the algorithm used to prove Theorem~\ref{thm:auxiliary_algorithm}.
%\begin{reptheorem}{thm:auxiliary_algorithm}
%There exists a polynomial time algorithm that given a vector $z \in [0, 1]^\cN$ and a value $t_s \in [0, 1]$ outputs a vector $x \in P$ obeying
%\begin{align*}
%	\bE[F(x)]
%	\geq
%	e^{t_s - 1} \cdot [(2 - t_s - e^{-t_s} - o(1)) \cdot f(OPT) &- (1 - e^{-t_s}) \cdot F(z \wedge \characteristic_{OPT}) \\&- (2 - t_s - 2e^{-t_s}) \cdot F(z \vee \characteristic_{OPT})]
%	\enspace.
%\end{align*}
%\end{reptheorem}
Proving the above theorem directly is made more involved by the fact that the vector $z$ might be fractional. Instead, we prove the following simplified version of Theorem~\ref{thm:auxiliary_algorithm} for integral values, and show that the simplified version implies the original one.

\begin{theorem} \label{thm:auxiliary_algorithm_simplified}
There exists a polynomial time algorithm that given a set $Z \subseteq \cN$ and a value $t_s \in [0, 1]$ outputs a vector $x \in P$ obeying
\begin{align*}
	\bE[F(x)]
	\geq
	e^{t_s - 1} \cdot [(2 - t_s - e^{-t_s} - o(1)) \cdot f(OPT) &- (1 - e^{-t_s}) \cdot f(Z \cap OPT) \\&- (2 - t_s - 2e^{-t_s}) \cdot f(Z \cup OPT)]
	\enspace.
\end{align*}
\end{theorem}

Next is the promised proof that Theorem~\ref{thm:auxiliary_algorithm_simplified} implies Theorem~\ref{thm:auxiliary_algorithm}.

\begin{proof}[Proof of Theorem~\ref{thm:auxiliary_algorithm} given Theorem~\ref{thm:auxiliary_algorithm_simplified}]
Consider an algorithm $ALG$ that given the $z$ and $t_s$ arguments specified by Theorem~\ref{thm:auxiliary_algorithm} executes the algorithm guaranteed by Theorem~\ref{thm:auxiliary_algorithm_simplified} with the same value $t_s$ and with a random set $Z$ distributed like $\RSet(z)$. The output of $ALG$ is then the output produced by the algorithm guaranteed by Theorem~\ref{thm:auxiliary_algorithm_simplified}. Let us denote this output by $x$.

Theorem~\ref{thm:auxiliary_algorithm_simplified} guarantees that, for every given $Z$,
\begin{align*}
	\bE[F(x) \mid Z]
	\geq
	e^{t_s - 1} \cdot [(2 - t_s - e^{-t_s} - o(1)) \cdot f(OPT) &- (1 - e^{-t_s}) \cdot f(Z \cap OPT) \\&- (2 - t_s - 2e^{-t_s}) \cdot f(Z \cup OPT)]
	\enspace.
\end{align*}
To complete the proof we take the expectation over $Z$ over the two sides of the last inequality and observe that
\[
	\bE[f(Z \cap OPT)]
	=
	\bE[f(\RSet(z) \cap OPT)]
	=
	\bE[f(\RSet(z \wedge \characteristic_{OPT}))]
	=
	F(z \wedge \characteristic_{OPT})
\]
and
\[
	\bE[f(Z \cup OPT)]
	=
	\bE[f(\RSet(z) \cup OPT)]
	=
	\bE[f(\RSet(z \vee \characteristic_{OPT}))]
	=
	F(z \vee \characteristic_{OPT})
	\enspace.
	\qedhere
\]
\end{proof}

%To further simplify our proof, we observe that when $f(OPT) \leq f(Z \cup OPT)$, the guarantee of Theorem~\ref

In the rest of this section we give a non-formal proof of Theorem~\ref{thm:auxiliary_algorithm_simplified}. This proof explains the main ideas necessary for proving the theorem, but uses some non-formal simplifications such as allowing a direct oracle access to the multilinear extension $F$ and giving the algorithm in the form of a continuous time algorithm (which cannot be implemented on a discrete computer). There are known techniques for getting rid of these simplifications (see, \eg,~\cite{CCPV11}), and a formal proof of Theorem~\ref{thm:auxiliary_algorithm_simplified} based on these techniques is given in Appendix~\ref{app:formal_proof_auxiliary}.

The algorithm we use for the non-formal proof of Theorem~\ref{thm:auxiliary_algorithm_simplified} is given as Algorithm~\ref{alg:aided_mcgreedy_continuous}. This algorithm starts with the empty solution $y(0) = \characteristic_\varnothing$ at time $0$, and grows this solution over time until it reaches the final solution $y(1)$ at time $1$. The way the solution grows varies over time. During the time range $[t_s, 1)$ the solution grows like in the Measured Continuous Greedy algorithm of~\cite{FNS11}. On the other hand, during the earlier time range of $[0, t_s)$ the algorithm pretends that the elements of $Z$ do not exist (by giving them negative marginal profits), and grows the solution in the way Measured Continuous Greedy would have grown it if it was given the ground set $\cN \setminus Z$.
The value $t_s$ is the time in which the algorithm switches between the two ways it uses to grow its solution, thus, the $s$ in the notation $t_s$ stands for ``switch''.

 %More specifically, at every time point $t \in [t_s, 1)$ the algorithm calculates for every element a weight based on the marginal contribution of increasing the coordinate of this element in the current solution to $1$, and then finds the vector $x(t) \in P$ which maximizes the linear function defined by these weights. Finally, the coordinates of the solution are increased in rates which depend on $x(t)$ and their current distance from the value $1$.

\begin{algorithm}
\caption{\textsf{Aided Measured Continuous Greedy (non-formal)}($f, P, Z, t_s$)} \label{alg:aided_mcgreedy_continuous}
\DontPrintSemicolon
Let $y(0) \gets \characteristic_\varnothing$.\\
\ForEach{$t \in [0, 1)$}
{
    For each $u \in \cN$ let $w_u(t) \gets F(y(t) \vee \characteristic_{u}) - F(y(t))$.\\
    Let $x(t) \gets \left\{\begin{array}{ll}\arg \max_{x\in P}\{\sum_{u\in \cN\setminus Z} w_u(t) \cdot x_u(t) - \sum_{u\in Z}x_u(t)\} & \mbox{if } t \in[0, t_s) \enspace, \\
     \arg \max_{x\in P}\left\{\sum_{u\in \cN} w_u(t) \cdot x_u(t)\right\} & \mbox{if } t \in[t_s,1) \enspace.\end{array}\right.$\\
%    If $t \leq t_s$, then let $x(t)\gets \arg \max_{x\in\cN \setminus Z}\{w(t) \cdot x(t)\}$.\\
%    Otherwise (if $t > t_s$), then let $x(t)\gets \arg \max_{x\in\cN}\{w(t) \cdot x(t)\}$.\\
	Increase $y(t)$ at a rate of $\frac{dy(t)}{dt} = (\characteristic_\cN - y(t)) \circ x(t)$.
}
\Return{$y(1)$.}
\end{algorithm}
%\footnotetext{This is can be done whenever $P$ is down-closed and solvable by defining
%\[
%	w'_u(t)
%	=
%	\begin{cases}
%		w_u(t) & \text{if $u \not \in Z$} \enspace,\\
%		-1 & \text{if $u \in Z$} \enspace,
%	\end{cases}
%\]
%and then finding the point $x(t) \in P$ maximizing $x(t) \cdot w'(t)$. Notice that the fact that $P$ is down-closed guarantees that the vector $x(t)$ found this way obeys $x_u(t) = 0$ for every $u \in Z$.}

We first note that algorithm outputs a vector in $P$.

\begin{observation} \label{obs:feasible}
$y(1) \in P$.
\end{observation}
\begin{proof}
Observe that $x(t)\in P$ at each time $t$, which implies that $(\characteristic_\cN - y(t)) \cdot x(t)$ is also in $P$ since $P$ is down-closed. Therefore, $y(1) = \int_0^1 (\characteristic_\cN - y(t)) \cdot x(t) dt$ is a convex combination of vectors in $P$, and thus, belongs to $P$.
\end{proof}

The following lemma lower bounds the increase in $F(y(t))$ as a function of $t$.

\begin{lemma} \label{lem:derivative_lower_bound_simple}
For every $t \in [0, 1)$,
\[
	\frac{dF(y(t))}{dt}
	\geq
	\begin{cases}
		F(y(t) \vee \characteristic_{OPT \setminus Z}) - F(y(t)) & \text{if $t \in [0, t_s)$} \enspace, \\
		F(y(t) \vee \characteristic_{OPT}) - F(y(t)) & \text{if $t \in [t_s, 1)$} \enspace.
	\end{cases}
\]
\end{lemma}
\begin{proof}
By the chain rule,
\begin{align} \label{eq:chain_rule}
	\frac{dF(y(t))}{dt}
	={} &
	\sum_{u \in \cN} \left(\frac{dy_u(t)}{dt} \cdot \left.\frac{\partial F(y)}{\partial y_u}\right|_{y = y(t)}\right)
	=
	\sum_{u \in \cN} \left((1 - y_u(t)) \cdot x_u(t) \cdot \left.\frac{\partial F(y)}{\partial y_u}\right|_{y = y(t)}\right)\\ \nonumber
	={} &
	\sum_{u \in \cN} \left(x_u(t) \cdot [F(y(t) \vee \characteristic_{u}) - F(y(t))]\right)
	=
	\sum_{u \in \cN} x_u(t) \cdot w_u(t)
	=
	x(t) \cdot w(t)
	\enspace.
\end{align}

Consider first the case $t \in [0, t_s)$.
During this time period Algorithm~\ref{alg:aided_mcgreedy_continuous} chooses $x(t)$ as the vector in $P$ maximizing $\sum_{u\in \cN\setminus Z} w_u(t) \cdot x_u(t) - \sum_{u\in Z}x_u(t)$. Since $P$ is down-closed $x(t)= \characteristic_{OPT \setminus Z}$ is in $P$ and has value $\characteristic_{OPT \setminus Z} \cdot w(t)$ and thus, we have $x(t) \cdot w(t) \geq \characteristic_{OPT \setminus Z} \cdot w(t)$. Plugging this observation into Equality~\eqref{eq:chain_rule} yields

%Algorithm~\ref{alg:aided_mcgreedy_continuous} chooses $x(t)$ as the vector in $P \wedge \characteristic_{\cN \setminus Z}$ maximizing $x(t) \cdot w(t)$. Since $\characteristic_{OPT \setminus Z} \in P \wedge \characteristic_{\cN \setminus Z}$, we must have $x(t) \cdot w(t) \geq \characteristic_{OPT \setminus Z} \cdot w(t)$. Plugging this observation into Equality~\eqref{eq:chain_rule} yields
\begin{align*}
	\frac{dF(y(t))}{dt}
	={} &
	x(t) \cdot w(t)
	\geq
	\characteristic_{OPT \setminus Z} \cdot w(t)
	=
	\sum_{u \in OPT \setminus Z} [F(y(t) \vee \characteristic_{u}) - F(y(t))]\\
	\geq{} &
	F(y(t) \vee \characteristic_{OPT \setminus Z}) - F(y(t))
	\enspace,
\end{align*}
where the last inequality holds by the submodularity of $f$.

Similarity, when $t \in [t_s, 1)$ Algorithm~\ref{alg:aided_mcgreedy_continuous} chooses $x(t)$ as the vector in $P$ maximizing $x(t) \cdot w(t)$. Since $\characteristic_{OPT} \in P$, we get this time $x(t) \cdot w(t) \geq \characteristic_{OPT} \cdot w(t)$. Plugging this observation into Equality~\eqref{eq:chain_rule} yields
\begin{align*}
	\frac{dF(y(t))}{dt}
	={} &
	x(t) \cdot w(t)
	\geq
	\characteristic_{OPT} \cdot w(t)
	=
	\sum_{u \in OPT} [F(y(t) \vee \characteristic_{u}) - F(y(t))]\\ \nonumber
	\geq{} &
	F(y(t) \vee \characteristic_{OPT}) - F(y(t))
	\enspace,
\end{align*}
where the last inequality holds again by the submodularity of $f$.
\end{proof}

\begin{lemma} \label{lem:bound1}
For every time $t \in [0, 1)$ and set $A\subseteq \cN$ it holds that %\Bnote{Can we remove the max with zero. It is not used anywhere (it will simplify in many places)??}
\[
	F(y(t) \vee \characteristic_{A})
	\geq
	\left(e^{-\max\{0,t-t_s\}}- e^{-t}\right)\max\left\{0, f(A)-f(A\cup Z)\right\} + e^{-t} \cdot f(A)
	\enspace.
\]
\end{lemma}

\begin{proof}
First, we note that for every time $t \in [0, 1]$ and element $u \in \cN$,
\begin{equation}\label{obs:bound}
	y_u(t) \leq
	\begin{cases}
		1 - e^{-t} & \text{if $u \not \in Z$} \enspace, \\
		1 - e^{-\max\{0,t-t_s\}} & \text{if $u \in Z$} \enspace.
	\end{cases}
\end{equation}

This follows for the following reason. Since $x(t)$ is always in $P\subseteq [0,1]^{\cN}$, %then for each $u\in \cN\setminus Z$,
$y_{u}(t)$ obeys the differential inequality
\[
	\frac{dy(t)}{dt}
	=
	(1 - y_u(t)) \cdot x(t) \leq (1-y_u(t))
	\enspace.
\]
Using the initial condition $y_u(0)=0$, the solution for this differential inequality is $y_u(t)\leq 1-e^{-t}$. To get the tighter bound for $u\in Z$, we note that at every time $t\in [0,t_s)$ Algorithm~\ref{alg:aided_mcgreedy_continuous} chooses as $x(t)$ a vector maximizing a linear function in $P$ which assigns a negative weight to elements of $Z$. Since $P$ is down-closed this maximum must have $x_u(t) = 0$ for every element $u \in Z$. This means that $y_u(t) = 0$ whenever $u \in Z$ and $t \in [0, t_s]$. Moreover, plugging the improved initial condition $y_u(t_s) = 0$ into the above differential inequality yields the promised tighter bound also for the range $(t_s, 1]$.

Next, let $\hat{f}$ be the Lov\'{a}sz extension of $f$. Then, by Lemma~\ref{lem:extensions_relation},
\begin{align}
	F(y(t) \vee \characteristic_{A})
	\geq{} & \hat{f}(y(t) \vee \characteristic_{A}) = \int_0^1 f(T_\lambda(y(t) \vee \characteristic_{A})) d\lambda  \nonumber\\
	\geq{} &
	\int_{1 - e^{-\max\{0,t-t_s\}}}^{1 - e^{-t}} f(T_\lambda(y(t) \vee \characteristic_{A})) d\lambda + \int_{1 - e^{-t}}^1 f(T_\lambda(y(t) \vee \characteristic_{A})) d\lambda \label{ineq10} \\
={} &
	\int_{1 - e^{-\max\{0,t-t_s\}}}^{1 - e^{-t}} f(T_\lambda(y(t) \vee \characteristic_{A})) d\lambda + e^{-t} \cdot f(A) \label{ineq11}\\
\geq & \left(e^{-\max\{0,t-t_s\}}- e^{-t}\right)\max\left\{0, f(A)-f(A\cup Z)\right\} + e^{-t} \cdot f(A) \label{ineq12}
	\enspace.
\end{align}
Inequality~\eqref{ineq10} follows by the non-negativity of $f$.
Equality~\eqref{ineq11} follows since, for $\lambda \in [1 - e^{-t}, 1)$, Inequality~\eqref{obs:bound} guarantees that $y_u(t) \leq \lambda$ for every $u \in \cN$, and thus, $T_\lambda(y(t) \vee \characteristic_{A}) = A$.
Finally Inequality~\eqref{ineq12} follows since, for $\lambda \in [1 - e^{-\max\{0,t-t_s\}}, 1 - e^{-t})$, Inequality~\eqref{obs:bound} guarantees that $y_u(t) \leq \lambda$ for every $u \in Z$, and thus, $T_\lambda(y(t) \vee \characteristic_{A}) = B(\lambda) \cup A$ for some $B(\lambda) \subseteq \cN \setminus Z$. By the non-negativity of $f$, $f(B(\lambda) \cup A)\geq 0$. Also, by the submodularity and non-negativity of $f$, for every such set $B(\lambda)$
\begin{align*}
	f(B(\lambda) \cup A)
	\geq{} &
	f(A) + f(B(\lambda) \cup Z \cup A) - f(Z \cup A)
\geq{}
	f(A) - f(Z \cup A)
	\enspace.
	\qedhere
\end{align*}
\end{proof}

Plugging the results of Lemma~\ref{lem:bound1} into the lower bound given by Lemma~\ref{lem:derivative_lower_bound_simple} on the improvement in $F(y(t))$ as a function of $t$ yields immediately the useful lower bound given by the next corollary.\footnote{Note that Corollary~\ref{cor:derivative_lower_bound} follows from a weaker version of Lemma~\ref{lem:bound1} which only guarantees $F(y(t) \vee \characteristic_{A}) \geq	(e^{-\max\{0,t-t_s\}}- e^{-t})\cdot[f(A)-f(A\cup Z)] + e^{-t} \cdot f(A)$. We proved the stronger version of the lemma above because it is useful in the formal proof of Theorem~\ref{thm:auxiliary_algorithm_simplified} given in Appendix~\ref{app:formal_proof_auxiliary}.}

\begin{corollary} \label{cor:derivative_lower_bound}
For every $t \in [0, 1)$,
\[
	\frac{dF(y(t))}{dt}
	\geq
	\begin{cases}
		f(OPT\setminus Z) - (1 - e^{-t}) \cdot f(Z \cup OPT) - F(y(t)) & \text{if $t \in [0, t_s)$} \enspace, \\
		e^{t_s-t} \cdot f(OPT) - (e^{t_s-t} - e^{-t}) \cdot f(Z \cup OPT) - F(y(t)) & \text{if $t \in [t_s, 1)$} \enspace.
	\end{cases}
\]
\end{corollary}

Using the last corollary we can complete the proof of Theorem~\ref{thm:auxiliary_algorithm_simplified}.

\begin{proof}[Proof of Theorem~\ref{thm:auxiliary_algorithm_simplified}]
We have already seen that $y(1)$---the output of Algorithm~\ref{alg:aided_mcgreedy_continuous}---belongs to $P$. It remains to show that 
\begin{align*}
	F(y(1))
	\geq
	e^{t_s - 1} \cdot [(2 - t_s - e^{-t_s}) \cdot f(OPT) &- (1 - e^{-t_s}) \cdot f(Z \cap OPT) \\&- (2 - t_s - 2e^{-t_s}) \cdot f(Z \cup OPT)]
	\enspace.
\end{align*}

%Clearly the lemma follows immediately from the non-negativity of $f$ when $f(OPT) = 0$. Thus, we may assume in the rest of the proof that $f(OPT) > 0$.

%We begin the proof by showing that the lemma is trivial when $f(OPT) - f(Z \cup OPT) + e^{-t_s} \cdot f(Z \cup OPT) \leq 0$. Notice that the last inequality implies
%\begin{align*}
	%(2 - t_s - e^{-t_s}) \cdot f(OPT)
	%\leq{} &
	%(2 - t_s - e^{-t_s}) \cdot (1 - e^{-t_s}) \cdot f(Z \cup OPT)\\
	%={} &
	%(2 - t_s - 3e^{-t_s} + t_se^{-t_s} + e^{-2t_s}) \cdot f(Z \cup OPT)\\
	%={} &
	%(2 - t_s - 2e^{-t_s}) \cdot f(Z \cup OPT) + e^{-t_s}(t_s + e^{-t_s} - 1) \cdot f(Z \cup OPT)
	%\enspace.
%\end{align*}
%
Corollary~\ref{cor:derivative_lower_bound} describes a differential inequality for $F(y(t))$. Given the boundary condition $F(y(0)) \geq 0$, the solution for this differential inequality within the range $t \in [0, t_s]$ is
\[
	F(y(t))
	\geq
	(1 - e^{-t}) \cdot f(OPT\setminus Z) - (1 - e^{-t} - te^{-t}) \cdot f(Z \cup OPT)
	\enspace.
\]
Plugging $t = t_s$ into the last inequality, we get
\[
	F(y(t_s))
	\geq
	(1 - e^{-t_s}) \cdot f(OPT\setminus Z) - (1 - e^{-t_s} - t_se^{-t_s}) \cdot f(Z \cup OPT)
	\enspace.
\]

Let $v=(1 - e^{-t_s}) \cdot f(OPT\setminus Z) - (1 - e^{-t_s} - t_se^{-t_s}) \cdot f(Z \cup OPT)$ be the right hand side of the last inequality.
Next, we solve again the differential inequality given by Corollary~\ref{cor:derivative_lower_bound} for the range $t \in [t_s, 1]$ with the boundary condition $F(y(t_s)) \geq v$. The resulting solution is
\begin{align*}
	F(y(t))
	\geq{} &
	e^{-t}\left[\left(t-t_s\right)\left( e^{t_s} \cdot f(OPT) - (e^{t_s}-1) \cdot f(Z \cup OPT)\right)+ve^{t_s}\right]
	\enspace
\end{align*}

Plugging $t=1$ and the value of $v$ we get
\begin{align}
	F(y(1))
	\geq{} &
	e^{-1}\left[\left(1-t_s\right)\left( e^{t_s} \cdot f(OPT) - (e^{t_s}-1) \cdot f(Z \cup OPT)\right)+ve^{t_s}\right] \nonumber\\
	\geq{} &
	\frac{1-t_s}{e}\left( e^{t_s} \cdot f(OPT) - (e^{t_s}-1) \cdot f(Z \cup OPT)\right) \label{inqq1} \\
	& +e^{t_s-1} \cdot \{(1 - e^{-t_s}) \cdot [f(OPT)-f(OPT\cap Z)] - (1 - e^{-t_s} - t_se^{-t_s}) \cdot f(Z \cup OPT)\} \nonumber\\
	={} & \begin{aligned}[t]e^{t_s - 1} \cdot [(2 - t_s - e^{-t_s}) \cdot f(OPT) &- (1 - e^{-t_s}) \cdot f(Z \cap OPT) \\&- (2 - t_s - 2e^{-t_s}) \cdot f(Z \cup OPT)] 	\enspace, \end{aligned}\nonumber
\end{align}
where Inequality \eqref{inqq1} follows since, by the submodularity and non-negativity of $f$,
\[
	f(OPT\setminus Z)
	\geq
	f(OPT)-f(OPT \cap Z) + f(\varnothing) \geq f(OPT)-f(OPT \cap Z)
	\enspace.
	\qedhere
\]

\end{proof}

%Theorem~\ref{thm:auxiliary_algorithm_simplified} now follows immediately from Lemma~\ref{lem:guarantee_in_terms_of_t_s}.

\bibliographystyle{plain}
\bibliography{Submodular}

\appendix
\section{A Formal Proof of Theorem~\ref{thm:auxiliary_algorithm_simplified}} \label{app:formal_proof_auxiliary}

In this section we give a formal proof of Theorem~\ref{thm:auxiliary_algorithm_simplified}. This proof is based on the same ideas used in the non-formal proof of this theorem in Section~\ref{sec:auxiliary_algorithm}, but employs also additional known techniques in order to get rid of the issues that make the proof from Section~\ref{sec:auxiliary_algorithm} non-formal.

The algorithm we use to prove Theorem~\ref{thm:auxiliary_algorithm_simplified} is given as Algorithm~\ref{alg:aided_mcgreedy}. This algorithm is a discrete variant of Algorithm~\ref{alg:aided_mcgreedy_continuous}. While reading the algorithm, it is important to observe that the choice of the values $\bar{\delta}_1$ and $\bar{\delta}_2$ guarantees that the variable $t$ takes each one of the values $t_s$ and $1$ at some point, and thus, the vectors $y(t_s)$ and $y(1)$ are well defined.

\begin{algorithm}
\caption{\textsf{Aided Measured Continuous Greedy}($f, P, Z, t_s$)} \label{alg:aided_mcgreedy}
\DontPrintSemicolon
\tcp{Initialization}
Let $\bar{\delta}_1 \gets t_s \cdot n^{-4}$ and $\bar{\delta}_2 \gets (1 - t_s) \cdot n^{-4}$.\\
Let $t \gets 0$ and $y(t) \gets \characteristic_\varnothing$.\\

\BlankLine

\tcp{Growing $y(t)$}
\While{$t < 1$}
{
	\ForEach{$u \in \cN$}{Let $w_u(t)$ be an estimate of $\bE[f(u \mid \RSet(y(t))]$ obtained by averaging the values of $f(u \mid \RSet(y(t))$ for $r = \lceil 48n^6 \ln (2n) \rceil$ independent samples of $\RSet(y(t))$.}
	Let $x(t) \gets \left\{\begin{array}{ll}\arg \max_{x\in P}\{\sum_{u\in \cN\setminus Z} w_u(t) \cdot x_u(t) - \sum_{u\in Z}x_u(t)\} & \mbox{if } t \in[0, t_s) \enspace, \\
     \arg \max_{x\in P}\left\{\sum_{u\in \cN} w_u(t) \cdot x_u(t)\right\} & \mbox{if } t \in[t_s,1) \enspace.\end{array}\right.$\\
	Let $\delta_t$ be $\bar{\delta}_1$ when $t < t_s$ and $\bar{\delta}_2$ when $t \geq t_s$.\\
	Let $y(t + \delta_t) \gets y(t) + \delta_t (\characteristic_\cN - y(t)) \circ x(t)$.\\
	Update $t \gets t + \delta_t$.
}

\BlankLine

\Return{$y(1)$.}
\end{algorithm}

We begin the analysis of Algorithm~\ref{alg:aided_mcgreedy} by showing that $y(t)$ remains within the cube $[0, 1]^\cN$ throughout the execution of the algorithm. Without this observation, the algorithm is not well-defined.% In the proof of this observation (and in most of the proofs that follow it), it is useful to denote by $\bar{\delta}_t$, for every time $t \in [0, 1)$, the value which is added to $t$ by Algorithm~\ref{alg:aided_mcgreedy} at the end of the iteration corresponding to this $t$. More formally, $\bar{\delta}_t = \delta_1$ when $t < t_s$, and $\bar{\delta}_t = \delta_2$ when $t \geq t_s$.

\begin{observation}
For every value of $t$, $y(t) \in [0, 1]^\cN$.
\end{observation}
\begin{proof}
We prove the observation by induction on $t$. Clearly the observation holds for $y(0) = \characteristic_\varnothing$. Assume the observation holds for some time $t$, then, for every $u \in \cN$,
\[
	y_u(t + \delta_t)
	=
	y_u(t) + \delta_t(1 - y_u(t)) \cdot x_u(t)
	\geq
	0
	\enspace,
\]
where the inequality holds since the induction hypothesis implies $1 - y_u(t) \in [0, 1]$. A similar argument also implies
\[
	y_u(t + \delta_t)
	=
	y_u(t) + \delta_t(1 - y_u(t)) \cdot x_u(t)
	\leq
	y_u(t) + (1 - y_u(t))
	=
	1
	\enspace.
	\qedhere
\]
\end{proof}

Using the last observation it is now possible to prove the following counterpart of Observation~\ref{obs:feasible}.

\begin{corollary} \label{cor:feasible_formal}
Algorithm~\ref{alg:aided_mcgreedy} always outputs a vector in $P$.
\end{corollary}
\begin{proof}
Let $T$ be the set of values $t$ takes during the execution of Algorithm~\ref{alg:aided_mcgreedy}. We observe that $\sum_{t \in T \setminus \{1\}} \delta_t = 1$, which implies that $\sum_{t \in T \setminus \{1\}} \delta_t \cdot x(t)$ is a convex combination of the vectors $\{x(t) : t \in T \setminus \{1\}\}$. As all these vectors belong to $P$, and $P$ is convex, any convex combination of them, including $\sum_{t \in T \setminus \{1\}} \delta_t \cdot x(t)$, must be in $P$.

Next, we rewrite the output of Algorithm~\ref{alg:aided_mcgreedy} as
\[
	y(1)
	=
	\sum_{t \in T \setminus \{1\}} \delta_t(\characteristic_\cN - y(t)) \circ x(t)
	\leq
	\sum_{t \in T \setminus \{1\}} \delta_t \cdot x(t)
	\enspace.
\]
By the above discussion the rightmost hand side of this inequality is a vector in $P$, which implies that $y(1) \in P$ since $P$ is down-closed.
\end{proof}

The next step towards showing that Algorithm~\ref{alg:aided_mcgreedy} proves Theorem~\ref{thm:auxiliary_algorithm_simplified} is analyzing its approximation ratio. We start this analysis by showing that with high probability all the estimations made by the algorithm are quite accurate. Let $\cA$ be the event that $|w_u(t) - \bE[f(u \mid \RSet(y(t)))]| \leq n^{-2} \cdot f(OPT)$ for every $u \in \cN$ and time $t$.

\begin{lemma}[The symmetric version of Theorem A.1.16 in~\cite{AS00}] \label{lem:concentration}
Let $X_i$, $1 \leq i \leq k$, be mutually independent with all $\bE[X_i] = 0$ and all $|X_i| \leq 1$. Set $S = X_1 + \dotsb + X_k$. Then, $\Pr[|S| > a] \leq 2e^{-a^2/2k}$.
\end{lemma}

\begin{corollary} \label{cor:A_prob}
$\Pr[\cA] \geq 1 - n^{-1}$.
\end{corollary}
\begin{proof}
Consider the calculation of $w_u(t)$ for a given $u \in \cN$ and time $t$. This calculation is done by averaging the value of $f(u \mid \RSet(y(t)))$ for $r$ independent samples of $\RSet(y(t))$. Let $Y_i$ denote the value of $f(u \mid \RSet(y(t)))$ obtained for the $i$-th sample, and let $X_i = \frac{Y_i - \bE[f(u \mid \RSet(y(t)))]}{2n \cdot f(OPT)}$. Then, by definition,
\[
	w_u(t)
	=
	\frac{\sum_{i = 1}^r Y_i}{r}
	=
	[2n \cdot f(OPT)] \cdot \frac{\sum_{i = 1}^r X_i}{r} + \bE[f(u \mid \RSet(y(t)))]
	\enspace.
\]

Since $Y_i$ is distributed like $f(u \mid \RSet(y(t)))$, the definition of $X_i$ guarantees that $\bE[X_i] = 0$ for every $1 \leq i \leq r$. Additionally, $|X_i| \leq 1$ for every such $i$ since the absolute values of both $Y_i$ and $\bE[f(u \mid \RSet(y(t)))]$ are upper bounded by $\max_{S \subseteq \cN} f(S) \leq n \cdot f(OPT)$ (the last inequality follows from our assumption that $\characteristic_u \in P$ for every element $u \in \cN$). Thus, by Lemma~\ref{lem:concentration},
\begin{align*}
	\Pr[|w_u(t) - \bE[f(u \mid \RSet(y(t)))]| > n^{-2} \cdot f(OPT)]
	={} &
	\Pr\left[\left|\sum_{i = 1}^r X_i\right| > \frac{r}{2n^3}\right]
	\leq
	2e^{-[rn^{-3}/2]^2/2r}\\
	={} &
	2e^{-rn^{-6}/8}
	\leq
	2e^{-6\ln (2n)}
	=
	2 \cdot \left(\frac{1}{2n}\right)^6
	\leq
	\frac{1}{2n^6}
	\enspace.
\end{align*}

Observe that Algorithm~\ref{alg:aided_mcgreedy} calculates $w_u(t)$ for every combination of element $u \in \cN$ and time $t < 1$. Since there are $n$ elements in $\cN$ and $2n^4$ times smaller than $1$, the union bound implies that the probability that for at least one such value $w_u(t)$ we have $|w_u(t) - \bE[f(u \mid \RSet(y(t)))]| > n^{-2} \cdot f(OPT)$ is upper bounded by
\[
	\frac{1}{2n^6} \cdot \left(n \cdot 2n^4\right)
	=
	\frac{1}{n}
	\enspace,
\]
which completes the proof of the corollary.
\end{proof}

Our next step is to give a lower bound on the increase in $F(y(t))$ as a function of $t$ given $\cA$. This lower bound is given by Corollary~\ref{cor:step_improvement}, which follows from the next two lemmata. The statement and proof of the corollary and the next lemma is easier with the following definition. Let $OPT'_t$ denote the set $OPT \setminus Z$ when $t < t_s$, and $OPT$ otherwise.

\begin{lemma} \label{lem:step_approximation}
Given $\cA$, for every time $t < 1$, $\sum_{u \in \cN} (1 - y_u(t)) \cdot x_u(t) \cdot \partial_u F(y(t))
\geq F(y(t) \vee \characteristic_{OPT'_t}) - F(y(t)) - O(n^{-1}) \cdot f(OPT)$.
\end{lemma}
\begin{proof}
Let us calculate the weight of $OPT'_t$ according to the weight function $w(t)$.
\begin{align*}
    w(t) \cdot \characteristic_{OPT'_t}
    ={} &
    \sum_{u \in OPT'_t} w_u(t)
    \geq
    \sum_{u \in OPT'_t} [\bE[f(u \mid \RSet(y(t)))] - n^{-2} \cdot f(OPT)] \\
    \geq{} &
    \mathbb{E}\left[\sum_{u \in OPT'_t} f(\RSet(y(t)) + u) - f(\RSet(y(t))) \right] - n^{-1} \cdot f(OPT)\\
    \geq{} &
    \mathbb{E}\left[f(\RSet(y(t)) \cup OPT'_t) - f(\RSet(y(t))) \right] - n^{-1} \cdot f(OPT)\\
    ={} &
    F(y(t) \vee \characteristic_{{OPT}'_t}) - F(y(t)) - n^{-1} \cdot f(OPT) \enspace,
\end{align*}
where the first inequality follows from the definition of $\cA$, and the last follows from the submodularity of $f$. Recall that $x(t)$ is the vector in $P$ maximizing some objective function (which depends on $t$). For $t < t_s$, the objective function maximized by $x(t)$ assigns the value $w(t) \cdot \characteristic_{OPT \setminus Z} = w(t) \cdot \characteristic_{OPT'_t}$ to the vector $\characteristic_{OPT'_t} \in P$. Similarly, for $t \geq t_s$, the objective function maximized by $x(t)$ assigns the value $w(t) \cdot \characteristic_{OPT} = w(t) \cdot \characteristic_{OPT'_t}$ to the vector $\characteristic_{OPT'_t} \in P$. Thus, the definition of $x(t)$ guarantees that in both cases we have
\[
    w(t) \cdot x(t)
		\geq
		w(t) \cdot \characteristic_{OPT'_t}
		\geq
		F(y(t) \vee \characteristic_{OPT'_t}) - F(y(t)) - n^{-1} \cdot f(OPT) \enspace.
\]
Hence,
\begin{align*}
    \sum_{u \in \cN} (1 - y_u(t)) \cdot x_u(t) \cdot &\partial_u F(y(t))
    =
    %\sum_{u \in \cN} (1 - y_u(t)) \cdot x_u(t) \cdot [F(y(t) \vee \characteristic_u) - F(y(t) \wedge \characteristic_{\bar{u}})]\\
    %={} &
    \sum_{u \in \cN} x_u(t) \cdot [F(y(t) \vee \characteristic_u) - F(y(t))]\\
		={} &
		\sum_{u \in \cN} x_u(t) \cdot \bE[f(u \mid \RSet(y(t)))]\\
		\geq{} &
		\sum_{u \in \cN} x_u(t) \cdot [w_u(t) - n^{-2} \cdot f(OPT)]
    =
    x(t) \cdot w(t) - n^{-1} \cdot f(OPT)\\
    \geq{} &
    F(y(t) \vee \characteristic_{OPT'_t}) - F(y(t)) - 2n^{-1} \cdot f(OPT)
    \enspace,
\end{align*}
where the first inequality holds by the definition of $\cA$ and the second equality holds since
\[
	F(y(t) \vee \characteristic_u) - F(y(t))
	=
	\bE[f(\RSet(y(t)) + u)] - \bE[f(\RSet(y(t)))]
	=
	\bE[f(u \mid \RSet(y(t)))]
	\enspace.
	\qedhere
\]
%The lemma now follows by plugging the definition of $OPT'_t$ into the lower bound we got on $\sum_{u \in \cN} (1 - y_u(t)) \cdot x_u(t) \cdot \partial_u F(y(t))$.
\end{proof}

\begin{lemma}[A rephrased version of Lemma~2.3.7 in~\cite{F13}] \label{lem:almost_linear}
Consider two vectors $x, x' \in [0, 1]^\cN$ such that $|x_u - x'_u| \leq \delta$ for every $u \in \cN$. Then, $F(x') - F(x) \geq \sum_{u \in \cN} (x'_u - x_u) \cdot \partial_u F(x) - O(n^3\delta^2) \cdot \max_{u \in N} f(\{u\})$.
\end{lemma}

\begin{corollary} \label{cor:step_improvement}
Given $\cA$, for every time $t < 1$, $F(y(t + \delta_t)) - F(y(t)) \geq \delta_t[F(y(t) \vee \characteristic_{OPT'_t}) - F(y(t))] - O(n^{-1} \delta_t) \cdot f(OPT)$.
\end{corollary}
\begin{proof}
Observe that for every $u \in \cN$, $|y_u(t + \delta_t) - y_u(t)| = |\delta_t(1 - y_u(t))x_u(t)| \leq \delta_t$. Hence, by Lemma~\ref{lem:almost_linear},
\begin{align}
	F(y(t + \delta_t)) - F(y(t))
	\geq{} &
	\sum_{u \in \cN} [y_u(t + \delta_t)) - y_u(t)] \cdot \partial_u F(y(t)) - O(n^3\delta_t^2) \cdot \max_{u \in N} f(\{u\}) \nonumber\\
	={} &
	\sum_{u \in \cN} \delta_t(1 - y_u(t)) \cdot x_u(t) \cdot \partial_u F(y(t)) - O(n^3\delta_t^2) \cdot \max_{u \in N} f(\{u\}) \label{eq:crude_improvement}
	\enspace.
\end{align}
Consider the rightmost hand side of the last inequality. By Lemma~\ref{lem:step_approximation}, the first term on this side can be bounded by
\begin{align*}
	\sum_{u \in \cN} \delta_t(1 - y_u(t)) \cdot x_u(t) \cdot \partial_u F(y(t))
	\geq{} &
	\delta_t \cdot [F(y(t) \vee \characteristic_{OPT'_t}) - F(y(t)) - O(n^{-1}) \cdot f(OPT)]\\
	={} &
	\delta_t \cdot [F(y(t) \vee \characteristic_{OPT'_t}) - F(y(t))] - O(n^{-1}\delta_t) \cdot f(OPT)
	\enspace.
\end{align*}
On the other hand, the second term of~\eqref{eq:crude_improvement} can be bounded by
\[
	O(n^3\delta_t^2) \cdot \max_{u \in N} f(\{u\})
	=
	O(n^{-1}\delta_t) \cdot f(OPT)
\]
since $\delta_t \leq n^{-4}$ by definition and $\max_{u \in N} f(\{u\}) \leq f(OPT)$ by our assumption that $\characteristic_u \in P$ for every $u \in \cN$.
\end{proof}

The lower bound given by the last corollary is in terms of $F(y(t) \vee \characteristic_{OPT'_t})$. To make this lower bound useful, we need to lower bound the term $F(y(t) \vee \characteristic_{OPT'_t})$. This is done by the following two lemma which corresponds to Lemma~\ref{lem:bound1}.

\begin{lemma} \label{lem:bound1_formal}[corresponds to Lemma~\ref{lem:bound1}]
For every time $t < 1$ and set $A\subseteq \cN$ it holds that
\begin{align*}
	F(y(t) \vee \characteristic_{A})
	\geq{} &
	\left(e^{-\max\{0,t-t_s\}}- e^{-t} - O(n^{-4})\right) \cdot \max\left\{0, f(A)-f(A\cup Z)\right\} \\&+ (e^{-t} - O(n^{-4})) \cdot f(A)
	\enspace.
\end{align*}
\end{lemma}

The proof of this lemma goes along the same lines as the proof of its corresponding lemma in Section~\ref{sec:auxiliary_algorithm}, except that the bounds on the coordinates of $y(t)$ used by the proof from Section~\ref{sec:auxiliary_algorithm} are replaced with the (slightly weaker) bounds given by the following lemma. %The slight weakness of this lemma compared to Observation~\ref{obs:y_formula} is the only reason Lemmata~\ref{lem:early_union_bound_formal} and~\ref{lem:late_union_bound_formal} are a bit weaker than their corresponding lemmata.

\begin{lemma} \label{lem:max_y}
For every time $t$ and element $u \in \cN$,
\[
	y_u(t) \leq
	\begin{cases}
		1 - e^{-t} + O(n^{-4}) & \text{if $u \not \in Z$} \enspace, \\
		1 - e^{-\max\{0,t-t_s\}} + O(n^{-4}) & \text{if $u \in Z$} \enspace.
	\end{cases}
\]
\end{lemma}
\begin{proof}
Let $\eps = n^{-4}$, and observe that $\delta_t \leq \eps$ for every time $t$. Our first objective is to prove by induction on $t$ that, if $y_u(\tau) = 0$ for some time $\tau \in [0, 1]$, then $y_u(t) \leq 1 - (1 - \eps)^{(t - \tau) / \eps}$ for every time $t \in [\tau, 1]$. For $t = \tau$ the claim holds because $y_u(\tau) = 0 = 1 - (1 - \eps)^{(\tau - \tau)/\eps}$. Next, assume the claim holds for some $t$, and let us prove it for $t + \delta_t$.
\begin{align*}
    y_u(t + \delta_t)
    &=
    y_u(t) + \delta_t (1 - y_u(t)) \cdot x_u(t)
		\leq
		y_u(t) + \delta_t (1 - y_u(t))
    =
    y_u(t) (1 - \delta_t ) + \delta_t\\
    &\leq
    (1 - (1 - \eps)^{(t - \tau) / \eps}) (1 - \delta_t) + \delta_t
    =
    %1 - (1 - \eps)^{t / \eps} + \bar{\delta}_t (1 - \eps)^{t / \eps}
    %\leq
    1 - (1 - \delta_t)(1 - \eps)^{(t - \tau) / \eps}\\
    &\leq
    1 - (1 - \eps)^{\delta_t / \eps}(1 - \eps)^{(t - \tau) / \eps}
		=
		1 - (1 - \eps)^{(t + \delta_t - \tau) / \eps}
		\enspace,
\end{align*}
where the last inequality holds since $(1 - x)^{1/x}$ is a decreasing function for $x \in (0, 1]$.

We complete the proof for the case $u \not \in Z$ by choosing $\tau = 0$ (clearly $y_u(0) = 0$) and observing that, for every time $t$,
\[
    1 - (1 - \eps)^{t / \eps}
    \leq
    1 - [e^{-1}(1 - \eps)]^t
    =
    1 - e^{-t}(1 - \eps)^t
    \leq
    1 - e^{-t}(1 - \eps)
    =
    1 - e^{-t} + O(\eps)
		\enspace.
\]

It remains to prove the lemma for the case $u \in Z$. Note that at every time $t\in [0,t_s)$ Algorithm~\ref{alg:aided_mcgreedy} chooses as $x(t)$ a vector maximizing a linear function in $P$ which assigns a negative weight to elements of $Z$. Since $P$ is down-closed this maximum must have $x_u(t) = 0$ for an element $u \in Z$. This means that $y_u(t) = 0$ for $t \in [0, t_s]$. In addition to proving the lemma for this time range, the last inequality also allows us to choose $\tau = t_s$, which gives, for $t \in [t_s, 1]$,
\[
	y_u(t)
	\geq
	1 - (1 - \eps)^{(t - t_s) / \eps}
	\geq
	1 - e^{t_s-t} + O(\eps)
	\enspace.
	\qedhere
\]
\end{proof}

Combining Corollary~\ref{cor:step_improvement} with Lemma~\ref{lem:bound1_formal} gives us the following corollary.

\begin{corollary} \label{cor:improvement}
Given $\cA$, for every time $t \in [0, t_s)$,
\begin{align*}
	F(y(t + \delta_t)) - F(y(t))
	\geq{} &
	\delta_t[f(OPT \setminus Z) - (1 - e^{-t}) \cdot f(Z \cup OPT) - F(y(t)))] \\&- O(n^{-1} \delta_t) \cdot f(OPT)
\end{align*}
and, for every time $t \in [t_s, 1)$,
\begin{align*}
	F(y(t + \delta_t)) - F(y(t))
	\geq{} &
	\delta_t[e^{-t} \cdot f(OPT) + (e^{t_s-t} - e^{-t}) \cdot \max\{f(OPT) - f(Z \cup OPT), 0\}\\ &-F(y(t))] - O(n^{-1} \delta_t) \cdot f(OPT)
	\enspace.
\end{align*}
\end{corollary}
\begin{proof}
For every time $t \in [0, t_s)$, Corollary~\ref{cor:step_improvement} and Lemma~\ref{lem:bound1_formal} imply together
\begin{align*}
	F(y(t + \delta_t)) - F(y(t))
	\geq{} &
	\delta_t[(1 - e^{-t} - O(n^{-4})) \cdot \max\left\{0, f(OPT \setminus Z)-f(OPT\cup Z)\right\} \\&+ (e^{-t} - O(n^{-4})) \cdot f(OPT \setminus Z)] - O(n^{-1} \delta_t) \cdot f(OPT)\\
	\geq{} &
	\delta_t[(1 - O(n^{-4})) \cdot f(OPT \setminus Z) - (1 - e^{-t}) \cdot f(Z \cup OPT) \\&- F(y(t)))] - O(n^{-1} \delta_t) \cdot f(OPT)
	\enspace.
\end{align*}
We observe that this inequality is identical to the inequality promised for this time range by the corollary, except that it has an extra term of $-\delta_t \cdot O(n^{-4}) \cdot f(OPT \setminus Z)$ on its right hand side. Since $f(OPT \setminus Z)$ is upper bounded by $f(OPT)$, due to the down-closeness of $P$, the absolute value of this extra term is at most
\[
	\delta_t \cdot O(n^{-4}) \cdot f(OPT)
	=
	O(n^{-1}\delta_t) \cdot f(OPT)
	\enspace,
\]
which completes the proof for the time range $t \in [0, t_s)$.

Consider now the time range $t \in [t_s, 1)$. For this time range Corollary~\ref{cor:step_improvement} and Lemma~\ref{lem:bound1_formal} imply together
\begin{align*}
	F(y(t + \delta_t)) - F(y(t))
	\geq{} &
	\delta_t[(e^{t_s-t}- e^{-t} - O(n^{-4})) \cdot \max\left\{0, f(OPT)-f(OPT\cup Z)\right\} \\&+ (e^{-t} - O(n^{-4})) \cdot f(OPT)] - O(n^{-1} \delta_t) \cdot f(OPT)
	\enspace.
\end{align*}
We observe again that this inequality is identical to the inequality promised for this time range by the corollary, except that it has extra terms of $-\delta_t \cdot O(n^{-4}) \cdot f(OPT)$ and $-\delta_t \cdot O(n^{-4}) \cdot \max\{0, f(OPT) - f(OPT \cup Z)\}$ on its right hand side. The corollary now follows since the absolute value of both these terms is upper bounded by $O(n^{-1}\delta_t) \cdot f(OPT)$.
\end{proof}

Corollary~\ref{cor:improvement} bounds the increase in $F(y(t))$ in terms of $F(y(t))$ itself. Thus, it gives a recursive formula which can be used to lower bound $F(y(t))$. Our remaining task is to solve this formula and get a closed-form lower bound on $F(y(t))$. Let $g(t)$ be defined as follows.	$g(0) = 0$ and for every time $t < 1$,
\begin{align*}
	g(t + &\delta_t)\\
	={} &
	\begin{cases}
		g(t) + \delta_t[f(OPT \setminus Z) - (1 - e^{-t}) \cdot f(Z \cup OPT) - g(t)] & \text{if $t < t_s$} \enspace,\\
		g(t) + \delta_t[e^{-t} \cdot f(OPT) + (e^{t_s-t} - e^{-t}) \cdot \max\{f(OPT) - f(Z \cup OPT), 0\} - g(t)] & \text{if $t \geq t_s$} \enspace.
	\end{cases}
\end{align*}
The next lemma shows that a lower bound on $g(t)$ yields a lower bound on $F(y(t))$.

\begin{lemma} \label{le:g_bound}
Given $\cA$, for every time $t$, $g(t) \leq F(y(t)) + O(n^{-1}) \cdot t \cdot f(OPT)$.
\end{lemma}
\begin{proof}
Let $c$ be the larger constant among the constants hiding behind the big $O$ notations in Corollary~\ref{cor:improvement}. We prove by induction on $t$ that $g(t) \leq F(y(t)) + (ct/n) \cdot f(OPT)$. For $t = 0$, this clearly holds since $g(0) = 0 \leq F(y(0))$. Assume now that the claim holds for some $t$, and let us prove it for $t + \delta_t$. There are two cases to consider. If $t < t_s$, then the induction hypothesis and Corollary~\ref{cor:improvement} imply, for a large enough $n$,
{\allowdisplaybreaks
\begin{align*}
    g(t + \delta_t)
    ={} &
    g(t) + \delta_t[f(OPT \setminus Z) - (1 - e^{-t}) \cdot f(Z \cup OPT) - g(t)]\\
    ={} &
    (1 - \delta_t)g(t) + \delta_t[f(OPT \setminus Z) - (1 - e^{-t}) \cdot f(Z \cup OPT)]\\
    \leq{} &
    (1 - \delta_t) [F(y(t)) + (ct/n) \cdot f(OPT)] + \delta_t[f(OPT \setminus Z) - (1 - e^{-t}) \cdot f(Z \cup OPT)]\\
    ={} &
    F(y(t)) + \delta_t[f(OPT \setminus Z) - (1 - e^{-t}) \cdot f(Z \cup OPT) - F(y(t))] \\&+ (ct/n) \cdot (1 - \delta_t) \cdot f(OPT)\\
    \leq{} &
    F(y(t + \delta_t)) + (c\delta_t/n) \cdot f(OPT) + (ct/n) \cdot (1 - \delta_t) \cdot f(OPT)\\
    \leq{} &
    F(y(t + \delta_t)) + [c(t + \delta_t)/n] \cdot f(OPT)
		\enspace.
\end{align*}
}

Similarly, if $t \geq t_s$, then we get
{\allowdisplaybreaks
\begin{align*}
    g(t + \delta_t)
    ={} &
    g(t) + \delta_t[e^{-t} \cdot f(OPT) + (e^{t_s-t} - e^{-t}) \cdot \max\{f(OPT) - f(Z \cup OPT), 0\} - g(t)]\\
    ={} &
    (1 - \delta_t)g(t) + \delta_t[e^{-t} \cdot f(OPT) + (e^{t_s-t} - e^{-t}) \cdot \max\{f(OPT) - f(Z \cup OPT), 0\}]\\
    \leq{} &
    (1 - \delta_t) [F(y(t)) + (ct/n) \cdot f(OPT)] \\&+ \delta_t[e^{-t} \cdot f(OPT) + (e^{t_s-t} - e^{-t}) \cdot \max\{f(OPT) - f(Z \cup OPT), 0\}]\\
    ={} &
    F(y(t)) + \delta_t[e^{-t} \cdot f(OPT) + (e^{t_s-t} - e^{-t}) \cdot \max\{f(OPT) - f(Z \cup OPT), 0\} - F(y(t))] \\&+ (ct/n) \cdot (1 - \delta_t) \cdot f(OPT)\\
    \leq{} &
    F(y(t + \delta_t)) + (c\delta_t/n) \cdot f(OPT) + (ct/n) \cdot (1 - \delta_t) \cdot f(OPT)\\
    \leq{} &
    F(y(t + \delta_t)) + [c(t + \delta_t)/n] \cdot f(OPT)
		\enspace.
		\qedhere
\end{align*}
}
\end{proof}

It remains to find a closed-form expression that lower bounds $g(t)$ (and thus, also $F(y(t))$). Let $h_1(t) \colon [0,t_s] \to \bR$ and $h_2(t) \colon [t_s, 1] \to \bR$ be defined as follows.
\[
	h_1(t)
	=
	(1 - e^{-t}) \cdot f(OPT \setminus Z) - (1 - e^{-t} - te^{-t}) \cdot f(Z \cup OPT)
	\enspace,
\]
and
\[
	h_2(t)
	=
	e^{-t} \cdot \{(t - t_s) \cdot [f(OPT) + (e^{t_s} - 1) \cdot \max\{f(OPT) - f(OPT \cup Z), 0\}] + e^{t_s} \cdot h_1(t_s)\}
	\enspace.
\]

\begin{lemma} \label{lem:h_1_bound}
For every time $t \leq t_s$, $h_1(t) \leq g(t)$.
\end{lemma}
\begin{proof}
The proof is by induction on $t$. For $t = 0$, $g(0) = 0 = (1 - e^0) \cdot f(OPT \setminus Z) - (1 - e^0 - 0 \cdot e^0) \cdot f(Z \cup OPT) = h_1(0)$. Assume now that the lemma holds for some $t < t_s$, and let us prove it holds also for $t + \delta_t$. By the induction hypothesis,
{\allowdisplaybreaks
\begin{align*}
    h_1(t + \delta_t)
    ={} &
    h_1(t) + \int_t^{t + \delta_t} h'(\tau) d\tau\\
    ={} &
    h_1(t) + \int_t^{t + \delta_t} \{e^{-\tau} \cdot f(OPT \setminus Z) - \tau e^{-\tau} \cdot f(Z \cup OPT)\} d\tau\\
    \leq{} &
    h_1(t) + \delta_t \cdot \{e^{-t} \cdot f(OPT \setminus Z) - t e^{-t} \cdot f(Z \cup OPT)\} d\tau\\
    ={} &
    (1 - \delta_t) h_1(t) + \delta_t \cdot \{f(OPT \setminus Z) - (1 - e^{-t}) \cdot f(Z \cup OPT)\}\\
    \leq{} &
    (1 - \delta_t) g(t) + \delta_t \cdot \{f(OPT \setminus Z) - (1 - e^{-t}) \cdot f(Z \cup OPT)\}
    =
    g(t + \delta_t) \enspace,
\end{align*}
}%
where the first inequality holds since $e^{-\tau}$ is a decreasing function of $\tau$ and $\tau e^{-\tau}$ is an increasing function of $\tau$ in the range $\tau \in [0, 1]$.% and $f(OPT) - f(Z \cap OPT) \geq 0$ by the definition of $P$ since the down-closure property of $P$ guarantees that $\characteristic_{Z \cap OPT} \in P$.
\end{proof}

\begin{lemma} \label{lem:h_2_bound}
For every time $t_s \leq t \leq 1$, $h_2(t) \leq g(t)$.
\end{lemma}
\begin{proof}
The proof is by induction on $t$. For $t = t_s$, by Lemma~\ref{lem:h_1_bound}, $h_2(t_s) = h_1(t_s) \leq g(t_s)$. Assume now that the lemma holds for some $t_s \leq t < 1$, and let us prove it holds also for $t + \delta_t$.

To avoid repeating complex expressions, let us denote $A = f(OPT) + (e^{t_s} - 1) \cdot \max\{f(OPT) - f(Z \cup OPT), 0\}$. Notice that $A$ is independent of $t$. Moreover, using this notation we can rewrite $h_2(t)$ as $h_2(t) = e^{-t} \cdot \{(t - t_s) \cdot A + e^{t_s} \cdot h_1(t_s)\}$. Thus, for every $\tau \in (t_s, 1)$,
\[
	h'_2(\tau)
	=
	-e^{-\tau} \cdot \{(\tau - t_s) \cdot A + e^{t_s} \cdot h_1(t_s)\} + e^{-\tau} \cdot A
	=
	e^{-\tau} \cdot \{(1 - \tau + t_s) \cdot A - e^{t_s} \cdot h_1(t_s)\}
	\enspace.
\]

The definition of $A$ and the non-negativity of $f$ imply immediately that $A \geq 0$. We would like to prove also that $t_s \cdot A - e^{t_s} \cdot h_1(t_s) \geq 0$. There are two cases to consider. First, if $f(OPT) \geq f(Z \cup OPT)$, then
{\allowdisplaybreaks
\begin{align*}
	t_s \cdot A - e^{t_s} \cdot h_1(t_s)
	={} &
	t_s \cdot f(OPT) + t_s(e^{t_s} - 1) \cdot \max\{f(OPT) - f(Z \cup OPT), 0\} \\&- (e^{t_s} - 1) \cdot f(OPT \setminus Z) + (e^{t_s} - 1 - t_s) \cdot f(Z \cup OPT)\\
	\geq{} &
	t_se^{t_s} \cdot f(OPT) - t_s(e^{t_s} - 1) \cdot f(Z \cup OPT) \\&- (e^{t_s} - 1) \cdot f(OPT) + (e^{t_s} - 1 - t_s) \cdot f(Z \cup OPT)\\
	={} &
	(t_se^{t_s} - e^{t_s} + 1) \cdot [f(OPT) - f(Z \cup OPT)]
	\geq
	0
	\enspace.
\end{align*}
}%
where the inequality uses the fact that $f(OPT) \geq f(OPT \setminus Z)$ because of the down-closure of $P$. On the other hand, if $f(OPT) < f(Z \cup OPT)$, then
\begin{align*}
	t_s \cdot A - e^{t_s} \cdot h_1(t_s)
	={} &
	t_s \cdot f(OPT) - (e^{t_s} - 1) \cdot f(OPT \setminus Z) + (e^{t_s} - 1 - t_s) \cdot f(Z \cup OPT)\\
	\geq{} &
	t_s \cdot f(OPT) -(e^{t_s} - 1) \cdot f(OPT) + (e^{t_s} - 1 - t_s) \cdot f(OPT)
	=
	0
	\enspace.
\end{align*}

Using the above observations and the induction hypothesis, we can now get
\begin{align*}
    h_2(t + \delta_t)
    ={} &
    h_2(t) + \int_t^{t + \delta_t} h'(\tau) d\tau
    =
    h_2(t) + \int_t^{t + \delta_t} e^{-\tau} \cdot \{(1 - \tau + t_s) \cdot A - e^{t_s} \cdot h_1(t_s)\} d\tau\\
    \leq{} &
    h_2(t) + \delta_t \cdot e^{-t} \cdot \{(1 - t + t_s) \cdot A - e^{t_s} \cdot h_1(t_s)\}
    =
    (1 - \delta_t) h_2(t) + \delta_t \cdot e^{-t} \cdot A\\
    \leq{} &
    (1 - \delta_t) g(t) + \delta_t \cdot e^{-t} \cdot A
    =
    g(t + \delta_t) \enspace.
		\qedhere
\end{align*}
\end{proof}

The last two lemmata give us the promised closed-form lower bound on $g(t)$, which can be used to lower bound the approximation ratio of Algorithm~\ref{alg:aided_mcgreedy}.

\begin{corollary} \label{cor:approximation_ratio}
$\bE[F(y(1))] \geq e^{t_s - 1} \cdot [(2 - t_s - e^{-t_s} - O(n^{-1})) \cdot f(OPT) - (1 - e^{-t_s}) \cdot f(Z \cap OPT) - (2 - t_s - 2e^{-t_s}) \cdot f(Z \cup OPT)]$.
\end{corollary}
\begin{proof}
By Lemma~\ref{le:g_bound}, given $\cA$,
\[
	F(y(1))
	\geq
	g(1) - O(n^{-1}) \cdot f(OPT)
	\enspace.
\]
By Lemma~\ref{lem:h_2_bound},
{\allowdisplaybreaks
\begin{align*}
	g(1)
	\geq{} &
	h_2(1)\\
	={} &
	e^{-1} \cdot \{(1 - t_s) \cdot [f(OPT) + (e^{t_s} - 1) \cdot \max\{f(OPT) - f(Z \cup OPT), 0\}] \\
	&+ (e^{t_s} - 1) \cdot f(OPT \setminus Z) - (e^{t_s} - 1 - t_s) \cdot f(Z \cup OPT)\}\\
	\geq{} &
	e^{-1} \cdot \{(1 - t_s) \cdot [e^{t_s} \cdot f(OPT) - (e^{t_s} - 1) \cdot f(Z \cup OPT)] \\
	&+ (e^{t_s} - 1) \cdot [f(OPT) - f(Z \cap OPT)] - (e^{t_s} - 1 - t_s) \cdot f(Z \cup OPT)\}\\
	={} &
	e^{t_s - 1} \cdot [(2 - t_s - e^{-t_s}) \cdot f(OPT) - (1 - e^{-t_s}) \cdot f(Z \cap OPT) \\&- (2 - t_s - 2e^{-t_s}) \cdot f(Z \cup OPT)]
	\enspace,
\end{align*}
}
where the second inequality holds since the submodularity and non-negativity of $f$ imply
\[
	f(OPT \setminus Z)
	\geq
	f(OPT) + f(\varnothing) - f(Z \cap OPT)
	\geq
	f(OPT) - f(Z \cap OPT)
	\enspace.
\]

Combining the above observations we get that, given $\cA$,
\begin{align*}
	F(y(1))
	\geq{} &
	e^{t_s - 1} \cdot [(2 - t_s - e^{-t_s} - O(n^{-1})) \cdot f(OPT) - (1 - e^{-t_s}) \cdot f(Z \cap OPT) \\&- (2 - t_s - 2e^{-t_s}) \cdot f(Z \cup OPT)]
	\enspace.
\end{align*}
Since $F(y(1))$ is always non-negative, this implies, by the law of total expectation,
\begin{align*}
	\bE[F(y(1))]
	\geq{} &
	\Pr[\cA] \cdot \{e^{t_s - 1} \cdot [(2 - t_s - e^{-t_s} - O(n^{-1})) \cdot f(OPT) - (1 - e^{-t_s}) \cdot f(Z \cap OPT) \\&- (2 - t_s - 2e^{-t_s}) \cdot f(Z \cup OPT)]\}\\
	\geq{} &
	\{e^{t_s - 1} \cdot [(2 - t_s - e^{-t_s} - O(n^{-1})) \cdot f(OPT) - (1 - e^{-t_s}) \cdot f(Z \cap OPT) \\&- (2 - t_s - 2e^{-t_s}) \cdot f(Z \cup OPT)]\} \\
	& - \frac{1}{n}\cdot e^{t_s - 1} \cdot (2 - t_s - e^{-t_s} - O(n^{-1})) \cdot f(OPT)\\
	={} &
	e^{t_s - 1} \cdot [(2 - t_s - e^{-t_s} - O(n^{-1})) \cdot f(OPT) - (1 - e^{-t_s}) \cdot f(Z \cap OPT) \\&- (2 - t_s - 2e^{-t_s}) \cdot f(Z \cup OPT)]
	\enspace,
\end{align*}
where the second inequality holds since $\Pr[\cA] \geq 1 - n^{-1}$ by Corollary~\ref{cor:A_prob}.
\end{proof}

Theorem~\ref{thm:auxiliary_algorithm_simplified} now follows immediately by combining Corollaries~\ref{cor:feasible_formal} and~\ref{cor:approximation_ratio}. 

\end{document}